\documentclass[a4paper,USenglish,autoref]{lipics-v2021}

\usepackage{algorithm}
\usepackage[noend]{algpseudocode}
\usepackage{alltt}
\usepackage{booktabs}
\usepackage{cite}
\usepackage{multirow}
\usepackage[autolanguage]{numprint} \npthousandsep{,}
\usepackage{pifont}  
\usepackage{textcomp}

\newif\ifcomment
\commenttrue     

\newcommand{\redcomment}[1]{\ifcomment\color{red} #1 \color{black}\fi}

\newif\iffigabbrv
\figabbrvfalse   
\newcommand{\figtext}{\iffigabbrv Fig.\else Figure\fi}

\newcommand{\bigo}[1]{\ensuremath{\mathcal{O}(#1)}}

\newcommand{\alg}{\ensuremath{\mathcal{A}}}
\newcommand{\msg}[2]{\ensuremath{\texttt{#1}(#2)}}
\newcommand{\countalg}{\textsc{Countdown}}
\newcommand{\Current}{\texttt{Current}}
\newcommand{\Maximum}{\texttt{Maximum}}

\makeatletter
\algrenewcommand\ALG@beginalgorithmic{\small}
\algrenewcommand\alglinenumber[1]{\footnotesize #1:}
\makeatother

\algdef{SE}[SUBALG]{Indent}{EndIndent}{}{\algorithmicend\ }%
\algtext*{Indent}
\algtext*{EndIndent}
\makeatletter
\newcommand{\multiline}[1]{%
  \begin{tabularx}{\dimexpr\linewidth-\ALG@thistlm}[t]{@{}X@{}}
    #1
  \end{tabularx}
}
\makeatother

\title{Memory Lower Bounds and Impossibility Results for Anonymous Dynamic Broadcast}
\titlerunning{Memory Lower Bounds and Impossibility for Anonymous Dynamic Broadcast}

\author{Garrett Parzych}{School of Computing and Augmented Intelligence \and Biodesign Center for Biocomputing, Security and Society\\Arizona State University, Tempe, AZ, USA}{gparzych@asu.edu}{}{}
\author{Joshua J. Daymude}{School of Computing and Augmented Intelligence \and Biodesign Center for Biocomputing, Security and Society\\Arizona State University, Tempe, AZ, USA}{jdaymude@asu.edu}{https://orcid.org/0000-0001-7294-5626}{}
\authorrunning{G.\ Parzych and J.\ J.\ Daymude}

\Copyright{Garrett Parzych and Joshua J.\ Daymude}



\funding{This work is supported in part by National Science Foundation award CCF-2312537.}

\hideLIPIcs
\nolinenumbers

\begin{document}

\maketitle

\begin{abstract}
    \textit{Broadcast} is a ubiquitous distributed computing problem that underpins many other system tasks.
    In static, connected networks, it was recently shown that broadcast is solvable without any node memory and only constant-size messages in worst-case asymptotically optimal time (Hussak and Trehan, PODC'19/STACS'20/DC'23).
    In the dynamic setting of adversarial topology changes, however, existing algorithms rely on identifiers, port labels, or polynomial memory to solve broadcast and compute functions over node inputs.
    We investigate \textit{space-efficient, terminating broadcast algorithms for anonymous, synchronous, 1-interval connected dynamic networks} and introduce the first memory lower bounds in this setting.
    Specifically, we prove that broadcast with \textit{termination detection} is impossible for idle-start algorithms (where only the broadcaster can initially send messages) and otherwise requires $\Omega(\log n)$ memory per node, where $n$ is the number of nodes in the network.
    Even if the termination condition is relaxed to \textit{stabilizing termination} (eventually no additional messages are sent), we show that any idle-start algorithm must use $\omega(1)$ memory per node, separating the static and dynamic settings for anonymous broadcast.
    This lower bound is not far from optimal, as we present an algorithm that solves broadcast with stabilizing termination using $\bigo{\log n}$ memory per node in worst-case asymptotically optimal time.
    In sum, these results reveal the necessity of non-constant memory for nontrivial terminating computation in anonymous dynamic networks.
\end{abstract}

\section{Introduction} \label{sec:intro}

Distributed algorithms for \textit{dynamic networks} enable processes to coordinate even as the communication links between them change over time, often rapidly and adversarially~\cite{Casteigts2012-timevaryinggraphs, Altisen2023-selfstabilizingsystems}.
When once a process disconnecting from a distributed system was viewed as a rare crash fault to tolerate, research over the last two decades has come to view dynamics as natural or even necessary to a system's function.
Application domains such as self-stabilizing overlay networks~\cite{Andersen2001-resilientoverlay, Feldmann2020-surveyalgorithms}, blockchains~\cite{Bessani2020-byzantinereplication, Gramoli2020-blockchainconsensus}, and swarm robotics~\cite{Hamann2018-swarmrobotics, Flocchini2019-distributedcomputing} are defined by their rapidly changing or physically moving components, forcing algorithms to achieve their goals by leveraging---or more often operating in spite of---these dynamics.

Despite the challenges, many fundamental problems have been addressed under adversarial dynamics, including broadcast, consensus, and leader election (see~\cite{Augustine2016-distributedalgorithmic, Casteigts2018-journeydynamic} for complementary surveys).
However, many of these algorithms endow their nodes with unique identifiers, port labels for locally distinguishing among their neighbors, (approximate) knowledge of the number of nodes in the network, or superlinear memories.
Taking inspiration from collective behavior in biological complex systems of computationally weak entities---such as foraging in ant colonies~\cite{Oh2023-adaptivecollective, Chandrasekhar2018-distributedalgorithm}, aggregation in slime mold spore migrations~\cite{Tarnita2015-fitnesstradeoffs}, and energy distribution in microbiomes~\cite{Prindle2015-ionchannels, Liu2015-metaboliccodependence}---we question to what extent these additional capabilities are necessary.
Specifically, we consider dynamic networks of \textit{anonymous} nodes (lacking both unique identifiers and port labels) with limited memory.
As an aside, anonymity is a desirable feature in its own right for engineering privacy-sensitive applications, such as Bluetooth-based contact tracing~\cite{Reichert2021-surveyautomatic}.

We consider the fundamental problem of synchronous \textit{broadcast}, in which all nodes in a network must eventually be informed of some information originating at a single node.
In static, connected networks, broadcast can be solved without any persistent memory at all: the \textit{amnesiac flooding} algorithm~\cite{Hussak2019-terminationflooding, Hussak2020-terminationflooding, Hussak2023-terminationamnesiac}---in which nodes forward copies of any message they receive to any neighbor that did not send them the message in the last round---informs all nodes (correctness) and allows them to eventually stop sending additional messages (stabilizing termination) within worst-case asymptotically optimal time.
But even when adversarial dynamics are constrained to maintain network connectivity in every round, a longstanding conjecture states that dynamic broadcast with stabilizing termination is impossible in $\bigo{\log n}$ memory without node identifiers or knowledge of $n$~\cite{ODell2005-informationdissemination}, let alone without any memory at all.
We summarize our contributions as follows.

\subparagraph{Our Contributions}

All results are proven with respect to deterministic algorithms run by anonymous nodes (lacking identifiers and port labels) in a dynamic network whose topology can change arbitrarily but remains connected in each synchronous round.
In this setting:
\begin{itemize}
    \item Broadcast with \textit{termination detection}---i.e., the broadcaster must eventually decide broadcast is complete---is impossible for idle-start algorithms where only the broadcaster can initially send messages (Section~\ref{sec:detection-impossibility}) and otherwise requires $\Omega(\log n)$ space (Section~\ref{sec:detection-bound}).

    \item Any idle-start algorithm solving broadcast with \textit{stabilizing termination}---i.e., eventually no additional messages are sent---must have $\omega(1)$ space complexity (Section~\ref{sec:stabilizing-bound}).
    We then present an algorithm solving broadcast with stabilizing termination in $\bigo{\log n}$ space and worst-case asymptotically optimal time (Section~\ref{sec:stabilizing-algo}).

\end{itemize}

We note that although synchronous systems are typically seen as less general than asynchronous ones, it is actually the opposite for the purposes of lower bounds.
Since an asynchronous adversary can always simulate a synchronous one, any impossibility results or lower bounds proven w.r.t.\ synchrony will apply to both types of systems.

\subsection{Model} \label{subsec:model}

\subparagraph{Dynamic Networks}

We consider a synchronous dynamic network comprising a fixed set of nodes $V$.
Nodes communicate with each other via message passing over a communication graph whose topology changes over time.
We model this topology as a \textit{time-varying graph} $\mathcal{G} = (V, E, T, \rho)$ where $V$ is the set of nodes, $E$ is the static set of undirected edges that may appear in the graph, $T = \mathbb{N}$ is the lifetime of the graph, and $\rho : E \times T \to \{0,1\}$ is the presence function indicating whether an edge exists at a given time~\cite{Casteigts2012-timevaryinggraphs}.
We refer to the set of edges present at time $t \in T$ as $E_t = \{e \in E: \rho(e, t) = 1\}$ and the undirected graph $G_t = (V, E_t)$ as the \textit{snapshot} of $\mathcal{G}$ at time $t \in T$.
We assume an adversary controls the presence function $\rho$ and that $E$ is the complete set of edges on $V$; i.e., we do not limit which edges the adversary can introduce.
We do, however, follow the majority of dynamic broadcast literature (e.g.,~\cite{Casteigts2012-timevaryinggraphs, ODell2005-informationdissemination, Kuhn2010-distributedcomputation, Kowalski2020-polynomialcounting, DiLuna2022-computinganonymous}) in assuming \textit{1-interval-connectivity} (also called ``always-connected snapshots''); i.e., the adversary may make arbitrary topological changes at each time $t \in T$ so long as each snapshot $G_t$ is connected.

\subparagraph{Node Capabilities}

Motivated by computationally weak individuals in biological collectives (e.g., cells, microbes, social insects, etc.), we consider nodes that are \textit{anonymous}, lacking unique identifiers, and have \textit{no knowledge or approximation of any global measure}, including the number of nodes $n$.
We further assume that nodes have \textit{no port labels}; i.e., they cannot count or locally distinguish among their neighbors.
Consequently, when a node communicates with its neighbors via \textit{message passing}, it does so using a broadcast mechanism, sending the same message to all its current neighbors.

\subparagraph{Algorithms and Execution}

Each node in the time-varying graph $\mathcal{G}$ synchronously executes the same distributed algorithm $\mathcal{A}$.
All nodes are initialized at time $t = 0$, and each synchronous round $t$ starting at time $t$ proceeds as follows:

\begin{enumerate}
    \item The adversary fixes the network topology $G_t$ for round $t$.

    \item Each node may send a message to its neighbors in $G_t$ according to algorithm $\alg$ as a function of its current state.

    \item Each node may perform a state transition according to algorithm $\alg$ as a function of its current state and the multiset of messages it (reliably) receives from its neighbors in $G_t$.
\end{enumerate}

\subparagraph{Memory} 

In this paper, we are primarily concerned with an algorithm's \textit{space complexity}, the maximum number of bits a node uses to store its state between rounds.
As usual for distributed systems, we are interested in the asymptotic growth of this measure as a function of $n = |V|$, the number of nodes.
We emphasize that even if nodes have $\Omega(\log n)$ memory---sufficient for storing unique identifiers---they are anonymous and are not assigned such identifiers a priori.
We also note that we do not analyze \textit{message complexity} directly, as our execution model specifies nodes that send messages based only on their states; thus, there are at most as many message types as states.

\subparagraph{Broadcast}

In the \textit{broadcast} problem, every node starts in the same state except for a single node known as the \textit{broadcaster} that is trying to deliver some information to every other node in the network.
We say that a node is \textit{informed} if it is the broadcaster or has previously received a message from an informed node.
All other nodes are \textit{uninformed}.
A broadcast is \textit{complete} when every node in the network is informed.
A distributed algorithm $\mathcal{A}$ \textit{solves} the broadcast problem in $t$ rounds if, for any time-varying graph whose nodes all execute $\mathcal{A}$, broadcast is completed by the end of round $t$.

\subparagraph{Idle-Start}

A node is \textit{idle} if it will not send a message in the subsequent round, and will not change its state if it does not receive a message (this has also been called ``quiescent''~\cite{Lynch1996-distributedalgorithms}, though that term is overloaded in this context).
Some broadcast algorithms critically rely on initializing all nodes as non-idle at time $t = 0$; conversely, an \textit{idle-start} algorithm initializes all nodes except the broadcaster as idle.

\subparagraph{Termination}

A simple solution to broadcast is to make every informed node continuously send messages.
Since we assume 1-interval connectivity, there is always at least one uninformed node receiving a message from an informed node in each round, so this algorithm solves broadcast in $\bigo{n}$ rounds.
This runtime bound is worst-case asymptotically optimal, but a smoothed analysis reveals significant improvements on more ``typical'' topologies~\cite{Dinitz2018-smoothedanalysis, Dinitz2022-smoothedanalysis}.
This algorithm also achieves $\Theta(1)$ space complexity, since nodes need only remember whether they're informed.
However, nodes sending messages forever creates undesirable congestion and precludes the system from advancing to further tasks, e.g., starting a new broadcast or using the broadcast information as part of a larger algorithm.
Thus, we seek algorithms meeting some kind of termination conditions.
An algorithm achieves \textit{stabilizing termination} if every node becomes idle within finite time.
An algorithm achieves the stronger condition of \textit{termination detection} if the broadcaster correctly and irrevocably decides that broadcast is complete (i.e., by entering a terminating state) within finite time.

\subsection{Related Work} \label{subsec:relwork}

Broadcast is a ubiquitous and well-studied distributed computing problem, often appearing as a building block in more complex tasks.
In static, connected networks, broadcast with stabilizing termination is solvable without any node memory (and thus without identifiers) and only $\Theta(1)$ message complexity in worst-case asymptotically optimal time, though port labels are required to distinguish among neighbors~\cite{Hussak2019-terminationflooding, Hussak2020-terminationflooding}.
In a recent extension of this work, the same algorithm was proven correct under node and edge deletion dynamics~\cite{Hussak2023-terminationamnesiac}, but breaks down under more general adversarial dynamics.
With this inspiration, our focus is space-efficient, terminating algorithms for broadcast in anonymous dynamic networks.

Early works on dynamic broadcast typically assumed stronger node capabilities.
A series of works on shortest, fastest, and foremost broadcast assumed local identifiers enabling a node $u$ to maintain a consistent label for any neighbor $v$, even if $v$ disconnected from and later reconnected to $u$~\cite{Casteigts2010-deterministiccomputations, Casteigts2012-timevaryinggraphs, Casteigts2015-shortestfastest}.
This assumption enables the construction of time-invariant logical structures like spanning trees, which---when combined with the assumption of recurrent dynamics (edges will eventually reappear)---reduces dynamic broadcast to static routing on these structures.
Similar techniques are used when assuming both unique node identifiers and shared knowledge of $n$, the number of nodes in the network~\cite{Raynal2014-simplebroadcast}.
Among these early works, only O'Dell and Wattenhofer~\cite{ODell2005-informationdissemination} share our focus on anonymous nodes and space complexity.
They conjectured that no algorithm can solve broadcast with stabilizing termination in $\bigo{\log n}$ space when nodes are anonymous and have no knowledge of $n$.
We build evidence against this conjecture with an algorithm achieving exactly this goal (Section~\ref{sec:stabilizing-algo}), though our synchronous execution model is not directly comparable to their combination of asynchronous time, bounded message latency, and disconnection detection for re-broadcasting messages.

A parallel line of work investigated what functions a dynamic network can deterministically compute over its nodes' inputs~\cite{Kuhn2010-distributedcomputation}.
In the context of anonymous dynamic networks, most results focus on the \textit{(exact) counting problem}~\cite{Michail2013-namingcounting, DiLuna2014-consciousunconscious, DiLuna2014-countinganonymous, DiLuna2016-nontrivial, Chakraborty2018-fasterexactcounting, Kowalski2020-polynomialcounting}, which terminating broadcast reduces to in $\bigo{n}$ time and $\bigo{\log n}$ space: once the broadcaster knows the number of nodes $n$, it need only wait $n$ rounds before every other node must have been informed (since the dynamic network is 1-interval connected), at which point it can terminate.
These works recently culminated in the exact characterization by Di Luna and Viglietta~\cite{DiLuna2022-computinganonymous, DiLuna2023-optimalcomputation} showing that anonymous dynamic networks with at least one leader can compute only the \textit{multi-aggregate functions}---those for which a node's output depends only on its own input and the multiset of all nodes' inputs---and do so in optimal (linear) time.
However, their algorithm uses $\Theta(n^3\log n)$ space in the worst case~\cite{DiLuna2023-briefannouncement}, leaving open what memory is necessary.
Our impossibility results and memory lower bounds for terminating broadcast (Sections~\ref{sec:detection-impossibility}--\ref{sec:stabilizing-bound}) shed light on this question, as many nontrivial multi-aggregate functions require information from at least one node to be communicated to all other nodes (e.g., minimums/maximums, averages, exact and generalized counting, etc.).

\section{Impossibility Results for Termination Detection}
\label{sec:detection-impossibility}

We begin by proving that there is no idle-start algorithm---i.e., one in which only the broadcaster can initially send messages---that solves broadcast with termination detection in our setting.
The idea behind the proof is as follows.
Supposing to the contrary that such an algorithm $\alg$ exists, it must solve broadcast and detect termination on any static, connected network $G$.
So we consider an execution of $\alg$ on an extension of $G$ as a time-varying graph $\mathcal{G}$ that we carefully construct to achieve two goals: (1) the execution of $\alg$ on $G$ in the time-varying graph $\mathcal{G}$ is identical to its execution on $G$ alone, and (2) there is an idle node in $\mathcal{G}$ that is sequestered from ever being informed.
This drives a contradiction: the broadcaster must detect termination in $\mathcal{G}$ because it does so on $G$ alone, but will do so incorrectly because there is still an uninformed node.

\begin{figure}[t]
    \centering
    \begin{subfigure}{0.3\textwidth}
        \centering
        \includegraphics[width=\textwidth]{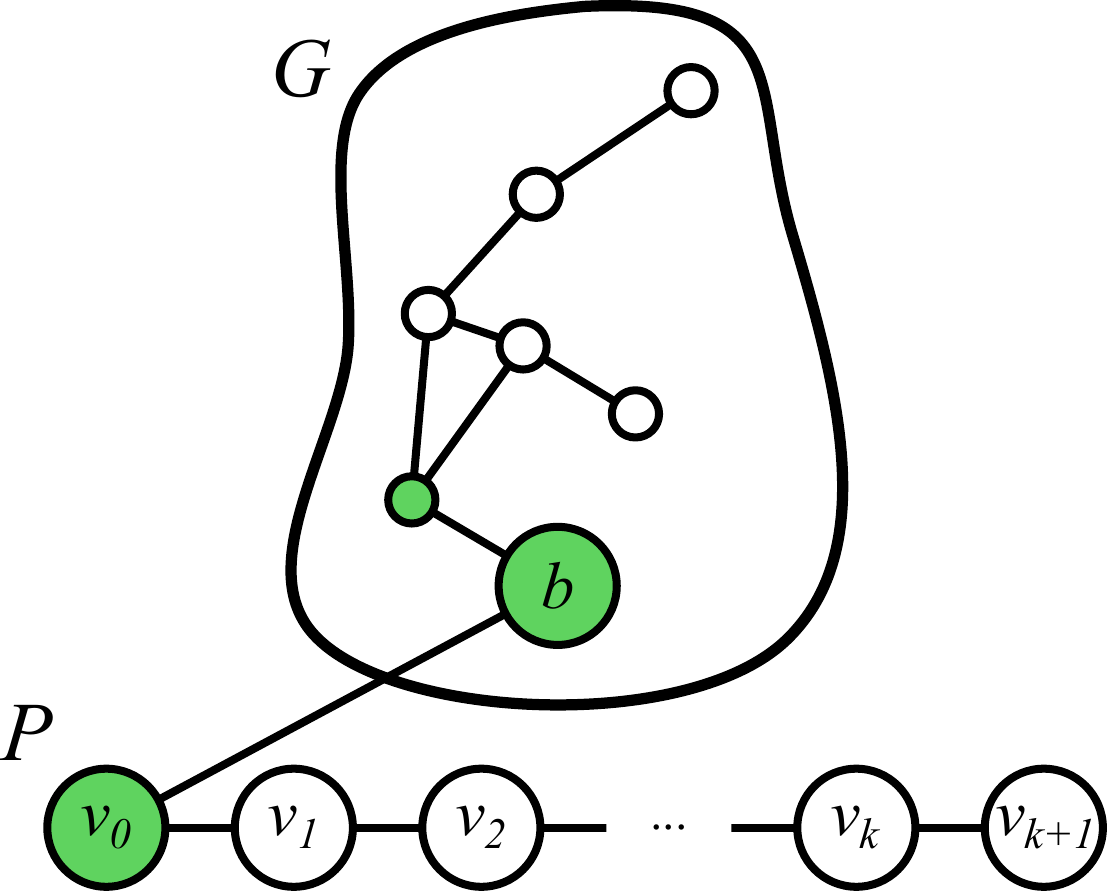}
        \caption{\centering End of Round 0}
        \label{fig:detect-impossible:0}
    \end{subfigure}
    \hfill
    \begin{subfigure}{0.3\textwidth}
        \centering
        \includegraphics[width=\textwidth]{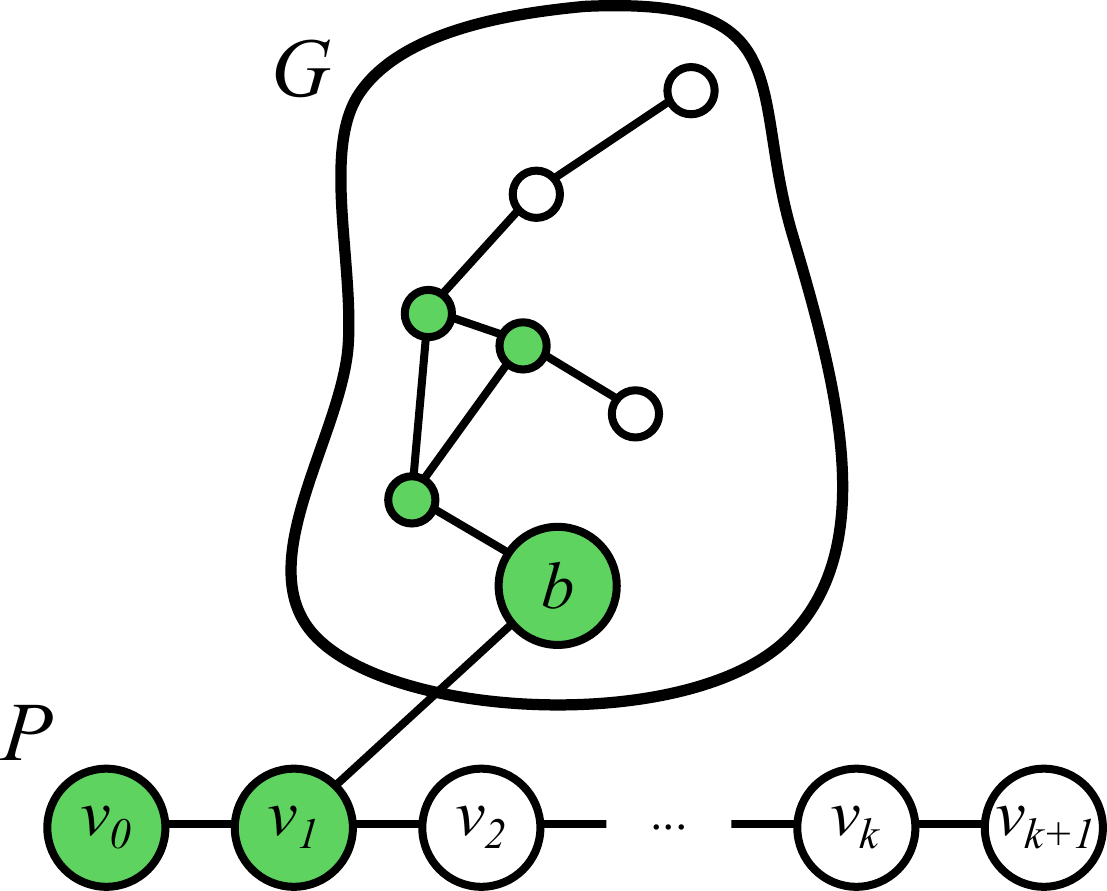}
        \caption{\centering End of Round 1}
        \label{fig:detect-impossible:1}
    \end{subfigure}
    \hfill
    \begin{subfigure}{0.3\textwidth}
        \centering
        \includegraphics[width=\textwidth]{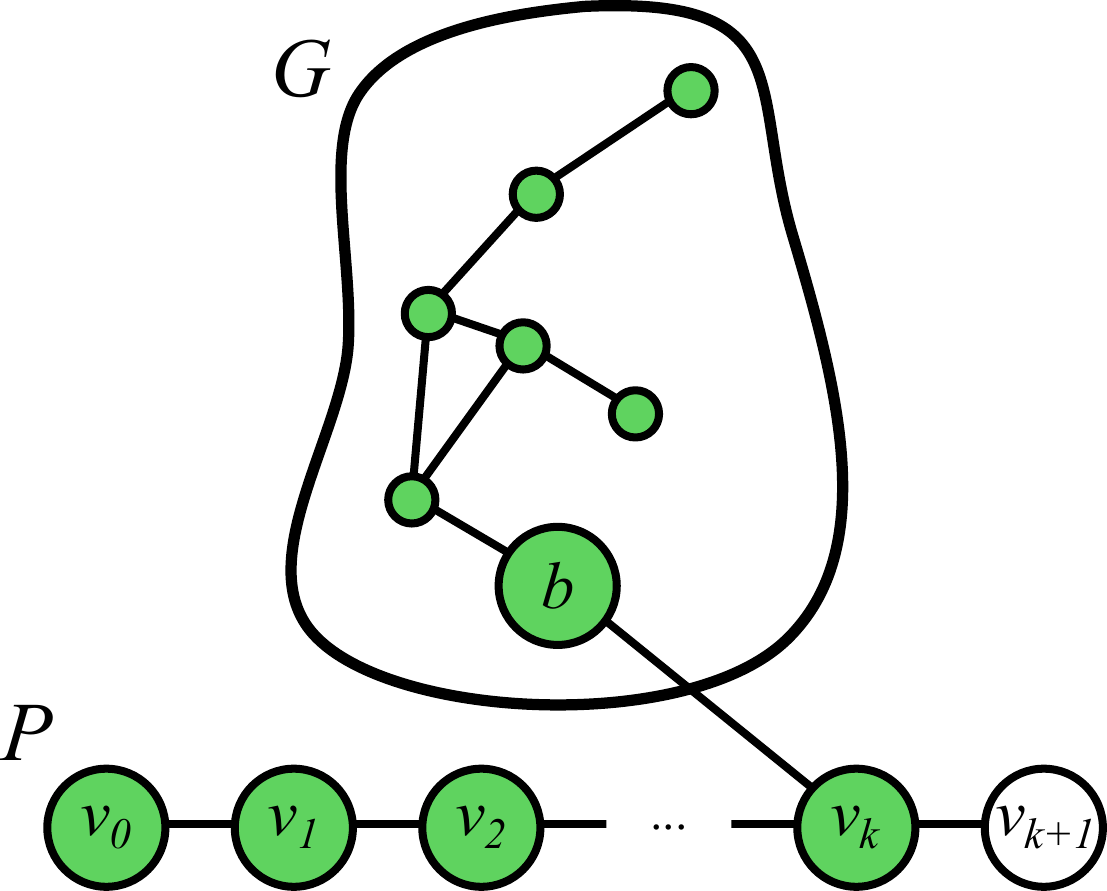}
        \caption{\centering End of Round $k$}
        \label{fig:detect-impossible:k}
    \end{subfigure}
    \caption{The time-varying graph $\mathcal{G}$ used in the proof of Theorem~\ref{thm:detection-impossibility}.
    In each round $t$, the broadcaster $b$ is connected to node $v_t$ in the path $P$.
    Informed nodes are shown in green.
    When $b$ declares broadcast to be complete in round $k$, node $v_{k+1}$ is still uninformed.}
    \label{fig:detect-impossible}
\end{figure}

\begin{theorem} \label{thm:detection-impossibility}
    No deterministic idle-start algorithm can solve broadcast with termination detection for anonymous, synchronous, 1-interval connected dynamic networks.
\end{theorem}
\begin{proof}
    Suppose to the contrary that there exists an idle-start algorithm $\alg$ solving broadcast with termination detection in our setting.
    Let $G$ be any static, connected graph and let $b$ be any node in $G$.
    If $\alg$ is executed on $G$ with $b$ as the broadcaster, there must exist some round $k$ during which $b$ correctly and irrevocably declares broadcast to be complete.

    Construct a time-varying graph $\mathcal{G}$ as follows (see \figtext~\ref{fig:detect-impossible}).
    First, $\mathcal{G}$ contains the static graph $G$ as a fixed part; i.e., all edges of $G$ will remain present throughout the lifetime of $\mathcal{G}$.
    Additionally, $\mathcal{G}$ contains a fixed path $P = v_0v_1 \cdots v_{k+1}$ of $k + 2$ nodes.
    Finally, in each round $t \in \{0, \ldots, k\}$, there is a single edge $\{b, v_t\}$ connecting the broadcaster $b$ to the path $P$.

    Consider the execution of $\alg$ on $\mathcal{G}$ with $b$ as the broadcaster.
    We argue by induction on $t \in \{0, \ldots, k\}$ that in round $t$, (1) all nodes in $G$ send the same messages and perform the same state transitions as they did in the execution of $\alg$ on $G$, and (2) the nodes $\{v_{t+1}, \ldots, v_{k+1}\}$ remain idle and uninformed.
    In round $t = 0$, only the broadcaster $b$ potentially sends messages and changes state since $\alg$ is an idle-start algorithm.
    Recall that nodes have no port labels and no knowledge of $n$.
    Thus, $b$ must send the same messages and perform the same state transition as in the execution of $\alg$ on $G$.
    The only neighbor of $b$ in the path $P$ is $v_0$, so nodes $\{v_1, \ldots, v_{k+1}\}$ receive no messages and remain idle and uninformed.

    Now suppose the claim holds up to and including some round $0 \leq t < k$.
    By the induction hypothesis, all nodes in $G$ have the same states at the start of round $t+1$ as in the execution of $\alg$ on $G$, and thus send the same messages.
    However, since $\{b, v_{t+1}\}$ is the only edge between $G$ and $P$ in round $t+1$, $b$ could in principle make a divergent state transition if it receives a message from $v_{t+1}$.
    But by the induction hypothesis, $v_{t+1}$ is idle and thus sends no messages in round $t+1$.
    For this same reason, nodes $\{v_{t+2}, \ldots, v_{k+1}\}$ receive no messages and remain idle and uninformed in round $t+1$.

    Therefore, in round $k$, the broadcaster $b$ irrevocably declares broadcast to be complete just as it did in the static setting, but $v_{k+1}$ remains uninformed, a contradiction.
\end{proof}

A related line of research investigates the computability of functions in dynamic networks with no or multiple leaders.
Di Luna and Viglietta have recently shown that if the number of leaders is known, nodes can compute the same functions as systems with a single leader. 
However, they found that if the number of leaders is unknown, then it is impossible for nodes to compute the size of the network with termination detection~\cite{DiLuna2023-optimalcomputation}.
Using similar ideas to our previous proof, we can show that broadcast with termination detection is also impossible without knowing the number of broadcasters.
Since broadcast can be reduced to counting the size of the network (simply broadcast for $n$ rounds after receiving the count), our theorem implies the result of Di Luna and Viglietta.
However, since we are unaware of any reduction from counting to broadcast, ours seems to be slightly more general.

\begin{theorem} \label{thm:detect-multileader}
    Even without an idle start, no deterministic algorithm can solve broadcast with termination detection for anonymous, synchronous, 1-interval connected dynamic networks if nodes have no knowledge of the number of broadcasters.
\end{theorem}
\begin{proof}
    Suppose for contradiction that an algorithm $\alg$ solves broadcast with termination detection without giving nodes knowledge of the number of broadcasters.
    Let the \textit{configuration} of a dynamic network at time $t$ be the multiset of node states at the start of round $t$.
    Let $C_0, C_1, \dots, C_x$ be the sequence of configurations that occur from running $\alg$ on the complete graph $K_3$ with a single broadcaster, where $C_x$ is the first configuration in which the broadcaster declares termination.
    We will create a new time-varying graph $\mathcal{G}$ on which $\alg$ will incorrectly terminate.
    First create $2^x$ copies of $C_0$, with each copy having a single broadcaster and two non-broadcasters.
    Add a path $p_0p_1 \dots p_x$ of $x+1$ nodes and for each copy of $C_0$, choose one of the non-broadcaster nodes and attach it to $p_0$; this will be the first snapshot of $\mathcal{G}$.

    Suppose that in $C_1$, the broadcaster was in state $\beta_1$ and the non-broadcasters in state $\alpha_1$.
    By symmetry of the graphs, after the first round of executing $\mathcal{A}$ on $\mathcal{G}$, each copy of $C_0$ will have a node in state $\beta_1$ and a node in state $\alpha_1$, while the node attached to $p_0$ will be in some other unknown state.
    Thus there are $2^x$ nodes in state $\beta_1$ and $2^x$ nodes in state $\alpha_1$.
    Use all the nodes in state $\alpha_1$ and half the nodes in state $\beta_1$ to create $2^{x-1}$ copies of the configuration $C_1$.
    Attach all of the unused nodes from the copies of $C_0$ to $p_1$.
    Again, choose a single non-broadcaster node from each copy of $C_1$, and attach them to $p_1$.
    Then run $\alg$ for one additional round.
    If $\beta_2$ and $\alpha_2$ are the states of the broadcaster and non-broadcasters in $C_2$ respectively, then by symmetry, our graph after this round will have $2^{x-1}$ copies each of $\beta_2$ and $\alpha_2$.
    We can now use these nodes to create $2^{x-2}$ copies of $C_2$.

    Continuing on in this way, we can create $2^{x-i}$ copies of configuration $C_i$ for each $i \in \{0, \ldots, x\}$ while only informing a single node in the path at a time.
    Thus, after $x$ rounds, we will have $2^{x - x} = 1$ copy of $C_x$ and $p_x$ will still be uninformed.
    But the broadcaster in the copy of $C_x$ will have declared termination, contradicting the correctness of $\mathcal{A}$.
\end{proof}

\section{Memory Lower Bound for Termination Detection} \label{sec:detection-bound}

In an idle-start algorithm, nodes have no indication of whether they have idle, uninformed neighbors.
This drives the indistinguishability result at the center of Theorem~\ref{thm:detection-impossibility}: a static, connected graph cannot tell if it's the entire network or a subgraph in a larger whole.
Without the constraints of an idle-start, however, this particular contradiction---and its corresponding impossibility result---disappears.
For example, the history tree algorithm of Di Luna and Viglietta solves broadcast in this setting in linear time and $\Theta(n^3\log n)$ space~\cite{DiLuna2023-briefannouncement}.
Still, all such algorithms must use at least logarithmic memory, as we now show.

\begin{figure}[t]
    \centering
    \includegraphics[width=\textwidth]{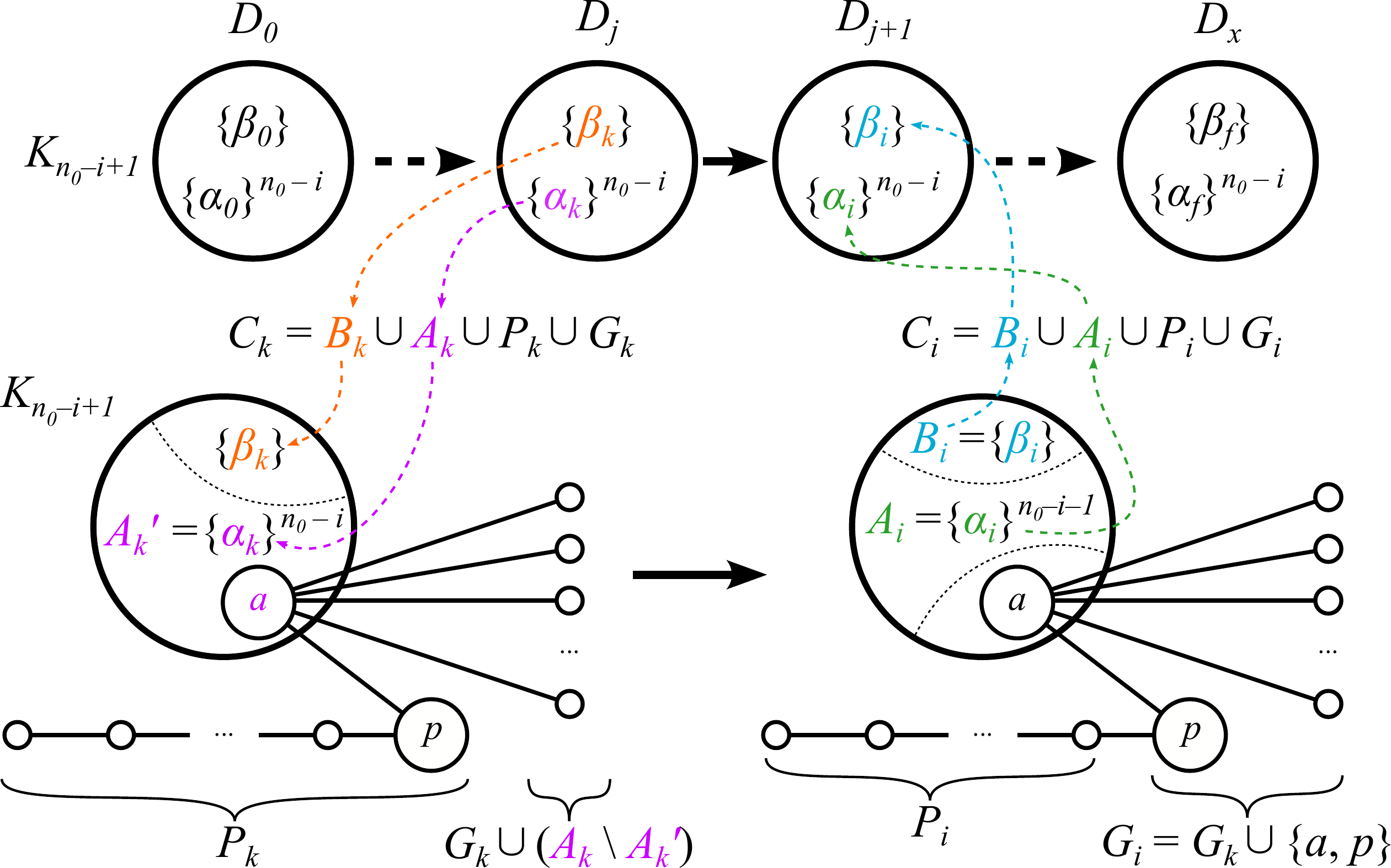}
    \caption{The inductive construction of configuration $C_i$ as described in the proof of Theorem~\ref{thm:detection-bound}.
    The key idea is to find a new pair of states $(\beta_i, \alpha_i)$ that is not already in $(\beta_0, \alpha_0), \ldots, (\beta_{i-1}, \alpha_{i-1})$ by arranging a previously identified reachable configuration $C_k$ and extending the corresponding execution of $\alg$ by one additional round.}
    \label{fig:detect-bound}
\end{figure}

\begin{theorem} \label{thm:detection-bound}
    Any algorithm that solves broadcast with termination detection for anonymous, synchronous, 1-interval connected dynamic networks must have $\Omega(\log n)$ space complexity.
\end{theorem}
\begin{proof}
    Consider any (non-idle-start) algorithm $\alg$ that solves broadcast with termination detection.
    Let $f(n)$ be the maximum number of states that $\alg$ uses when run on dynamic networks of at most $n$ nodes.
    We will show that $f(n) \geq n^{1/2}$ for all $n \geq 1$, implying that $\alg$ uses $\log(f(n)) \geq \log(n^{1/2}) = \Omega(\log n)$ space.
    
    Suppose to the contrary that there exists an $n_0 \geq 1$ such that $f(n_0) < n_0^{1/2}$.
    A configuration $C$ is \textit{reachable} (from an initial configuration) if there exists a time-varying graph $\mathcal{G}$ and time $t$ such that the execution of $\mathcal{A}$ on $\mathcal{G}$ for $t$ rounds results in configuration $C$.
    We will find a sequence of $n_0$ reachable configurations $(C_i)_{i=0}^{n_0 - 1}$ where each $C_i$ is a disjoint union of multisets\footnote{A multiset $X$ is a disjoint union of multisets $Y \cup Z$ if the multiplicity of any element $x \in X$ is the sum of multiplicities of $x$ in $Y$ and $Z$.} of states $B_i \cup A_i \cup P_i \cup G_i$ satisfying:
    
    \begin{enumerate}
    	\item $\lvert C_i \rvert = 2n_0$,
    	\item $B_i = \{\beta_i\}$, where $\beta_i$ is the state of the broadcaster,
    	\item $A_i = \{\alpha_i\}^{n_0 - i - 1}$, exactly $n_0 - i - 1$ copies of the same state $\alpha_i$,
    	\item $\lvert P_i \rvert \geq n_0 - i$ and contains only uninformed states,
    	\item $G_i = C_i \setminus (B_i \cup A_i \cup P_i)$, and
    	\item $(\beta_i, \alpha_i) \neq (\beta_j, \alpha_j)$ for all $j \neq i$.
    \end{enumerate}

    Initially, the broadcaster is in some state $\beta_0$ and all other nodes are uninformed in some state $\alpha_0 \neq \beta_0$.
    Define the initial configuration $C_0$ by letting $B_0 = \{\beta_0\}$, $A_0 = \{\alpha_0\}^{n_0 - 1}$, $P_0 = \{\alpha_0\}^{n_0}$, and $G_0 = \varnothing$.
    Clearly, $C_0$ is reachable and meets the above conditions.
    
    Now consider any $1 \leq i < n_0$ and suppose that $C_0, \ldots, C_{i-1}$ have already been defined; we define $C_i$ recursively as follows (see \figtext~\ref{fig:detect-bound} for an illustration).
    Let $\mathcal{D} = (D_0, \ldots, D_x)$ be the sequence of configurations obtained by running $\mathcal{A}$ on $K_{n_0 - i + 1}$, the complete graph on $n_0 - i + 1$ nodes, where $D_x$ is the first configuration in which the broadcaster declares termination.
    Note that, by the symmetry of the complete graph, each configuration in $\mathcal{D}$ has one state for the broadcaster and one state shared by every other node.
    Let $j \in \{0, \ldots, x\}$ and $k \in \{0, \ldots, i-1\}$ be such that $\beta_k \in D_j$ and $D_j \setminus \{ \beta_k \} \subseteq A_k$ and for any other $j', k'$ fulfilling this condition, $j' \leq j$.
    Note that such a $j$ must exist since, firstly, $D_0$ is the initial configuration of $K_{n_0 - i + 1}$ and thus $\beta_0 \in D_0$ and $D_0 \setminus \{\beta_0\} = \{\alpha_0\}^{n_0 - i} \subseteq \{\alpha_0\}^{n_0 - 1} = A_0$; and secondly, the sequence $\mathcal{D}$ is finite since $\mathcal{A}$ terminates in finite time.
    We also have $j < x$ since otherwise there exists a configuration $C_k$ containing $|P_k| \geq n_0 - k > 0$ uninformed nodes (by induction), but $\beta_k \in D_x$ is a terminating state, contradicting the correctness of $\alg$.
    
    Configuration $C_k$ is reachable, so there is some time-varying graph upon which the execution of $\mathcal{A}$ will within finite time be in configuration $C_k = B_k \cup A_k \cup P_k \cup G_k$.
    We extend this execution by putting this graph into the following topology and executing $\mathcal{A}$ for a single additional round.
    In a slight abuse of notation, we refer here to configurations as sets of nodes instead of multisets of states since edge dynamics make nodes in the same state interchangeable.
    Choose any subset $A_k' \subseteq A_k$ of $n_0 - i$ nodes (which is well-defined since $|A_k| = n_0 - k - 1 \geq n_0 - i$ by induction) and arrange them with $B_k$ as a complete graph $K_{n_0 - i + 1}$.
    Arrange the nodes in $P_k$ as a path.
    Connect these components by attaching an end $p \in P_k$ of the path to some node $a \in A_k'$ and connect every node in $G_k \cup (A_k \setminus A_k')$ only to this node $a$.
    Let $C_i$ be the configuration obtained after executing one round of $\alg$ on this topology.
    Clearly $C_i$ is reachable; we show next that it satisfies the required conditions.
    
    \begin{enumerate}
    	\item $\lvert C_i \rvert = \lvert C_k \rvert = 2n_0$ as desired.
    	
        \item Let $B_i = \{\beta_i\}$ be the state of the unique node in $B_k$ after this one additional round.
    	
        \item Consider any two nodes $a_1, a_2 \in A_k' \setminus \{a\}$ before execution of the round.
        Notice that they are both connected to $B_k$ and every node in $A_k'$ except for themselves.
        By induction, every node in $A_k'$ has the same state $\alpha_k$.
        Thus, $a_1$ and $a_2$ receive the same sets of messages from their neighbors, so they transition to the same state in the next round.
        Let $\alpha_i$ be this state shared by all nodes in $A_k' \setminus \{a\}$ after the round's execution and let $A_i$ be the corresponding configuration.
        Thus $A_i = \{\alpha_i\}^{|A_k'| - 1} = \{\alpha_i\}^{n_0 - i - 1}$ as desired.
    	
        \item Consider any node $v \in P_k \setminus \{p\}$ before execution of the round.
        By induction, $v$ and its neighbors are uninformed.
        Thus, $v$ receives no message from an informed node in the subsequent round and remains uninformed.
        Let $P_i$ be the (states of) nodes $P_k \setminus \{p\}$.
        By induction, we have $\lvert P_i \rvert = |P_k| - 1 \geq n_0 - k - 1 \geq n_0 - i$ as desired.
        
        \item $G_i = C_i \setminus (B_i \cup A_i \cup P_i)$ is simply defined as all the remaining nodes from $C_k$.
        
        \item Suppose for contradiction that $(\beta_i, \alpha_i) = (\beta_\ell, \alpha_\ell)$ for some $\ell \in \{0, \dots, i-1\}$.
        Recall that the configuration $D_j = \{\beta_k\} \cup \{\alpha_k\}^{n_0 - i}$ was obtained by running $\alg$ on $K_{n_0 - i + 1}$ and $D_j$ is not the final configuration of that execution.
        Moreover, observe that when we extended the execution reaching $C_k$ by one round, we reconstructed this exact configuration and topology by arranging $A_k' \cup B_k$ as a complete graph.
        When running $\alg$ for one additional round on this component, the broadcaster in $D_j$ and $B_k$ transitions from $\beta_k$ to $\beta_i$ and the non-broadcasters in $D_j$ and $A_k' \setminus \{a\}$ transition from $\alpha_k$ to $\alpha_i$.
        Thus, $D_{j+1} = \{\beta_i\} \cup \{\alpha_i\}^{n_0 - i} = \{\beta_\ell\} \cup \{\alpha_\ell\}^{n_0 - i}$, by supposition.
        But this implies $\beta_\ell \in D_{j+1}$ and $D_{j+1} \setminus \{\beta_\ell\} = \{\alpha_\ell\}^{n_0 - i} \subseteq \{\alpha_\ell\}^{n_0 - \ell} = A_\ell$, contradicting the maximality of $j$.
    \end{enumerate}

    Thus, we obtain the desired sequence of $n_0$ reachable configurations $(C_i)_{i=0}^{n_0 - 1}$ and their corresponding distinct state pairs $(\beta_i, \alpha_i)$.
    The initial state pair $(\beta_0, \alpha_0)$ appears in every execution of $\alg$ and the remaining state pairs $(\beta_i, \alpha_i)$ for $i \geq 1$ appear in executions of $\alg$ on complete graphs $K_{n_0 - i + 1}$ of at most $n_0$ nodes.
    But the maximum number of states $\alg$ can use on dynamic networks of at most $n_0$ nodes is $f(n_0) < n_0^{1/2}$, so there are only $f(n_0)^2 < n_0$ distinct state pairs, a contradiction.
\end{proof}

\section{Memory Lower Bound for Stabilizing Termination from Idle-Start} \label{sec:stabilizing-bound}

We now turn our attention from termination detection to stabilizing termination, requiring only that all nodes are eventually informed and stop sending messages.\
In this section, we prove that any idle-start algorithm solving broadcast with stabilizing termination must use superconstant memory.
This shows that stabilizing broadcast is strictly harder in dynamic networks than in the static setting, which has a trivial constant-memory algorithm (if uninformed and receiving a message, become informed and forward the message to all neighbors) and even permits an algorithm with no persistent memory at all~\cite{Hussak2019-terminationflooding,Hussak2020-terminationflooding}.

Our proof will use similar time-varying graphs as in the proofs of Theorems~\ref{thm:detection-impossibility} and~\ref{thm:detection-bound}.
However, those proofs derived contradictions from the broadcaster declaring termination too early, a condition that cannot be used in the case of stabilizing termination.
Instead, we suppose a constant memory algorithm exists and use it to create an infinite sequence of configurations satisfying some very specific properties.
We show in Lemma~\ref{lem:alpha}, however, that no infinite sequence with these properties exists.

In Lemma~\ref{lem:f}, we show that the time an algorithm takes to stabilize on a static complete graph, starting from any reachable configuration, depends only on the number of states used and not on the number of nodes in the graph.
This will be useful in our proof since we consider algorithms using constant memory, thus implying a fixed bound on the stabilization time for any complete graph, regardless of size.
To this end, we define the following function.

\begin{definition} \label{def:f}
    Let $f(k)$ be the minimum number of rounds such that for any idle-start algorithm $\alg$ solving broadcast with stabilizing termination using at most $k$ states and any configuration $C$ reachable by $\alg$, an execution of $\alg$ on a static complete graph starting in configuration $C$ stabilizes within $f(k)$ rounds.
\end{definition}

We first prove that this function is well-defined and bounded.

\begin{lemma} \label{lem:f}
	$f(k) \leq 3k!$
\end{lemma}
\begin{proof}
    Let $\alg$ be an idle-start algorithm solving broadcast with stabilizing termination using only $k$ states and consider its execution on a static complete graph starting from some reachable configuration $C$. 
	At each round, consider partitioning the nodes into $\ell \leq k$ sets by their current state.
    In each round, any nodes with the same state receive the same messages from their neighbors and transition to the same next state, so nodes sharing a state continue to do so throughout the execution.
	Thus, although the number of sets $\ell$ can decrease over time as some nodes converge to the same state, it can never increase.
	If there are $\ell$ sets of nodes with distinct states at one time, there are $\binom{k}{\ell}$ possibilities for what these states are and $\ell!$ ways for these states to be assigned to the sets.
    Again, since the number of of sets can decrease, $\ell$ may take on any value in $\{1, \ldots, k\}$ throughout the execution.

    Since $\alg$ is deterministic, if the same assignment of states to the same sets ever occurs twice in an execution, the algorithm must be in a loop.
    However, this can not happen since $C$ is reachable and thus $\alg$ must stabilize in finite time.
	Thus, the maximum number of configurations $\alg$ can visit before stabilizing is
	\[\sum_{\ell=1}^k \binom{k}{\ell} \cdot \ell!
    = \sum_{\ell=1}^k \frac{k!}{(k-\ell)!}
    = k! \cdot \sum_{i=0}^{k-1} \frac{1}{i!}
    \leq k! \cdot e \qedhere\]
\end{proof}

Our next lemma sets us up for the contradiction in Theorem~\ref{thm:stable-bound}.

\begin{lemma} \label{lem:alpha}
	Let $\mathcal{S} = \{S_1, S_2, \dots\}$ be a collection of multisets $S_i$ which are each the disjoint union of multisets $A_i \cup B_i$ with elements from $[n] = \{1, \dots, n\}$ satisfying $\lvert A_i \rvert = k$, $\lvert B_i \rvert$ is finite, and $\lnot (A_j = A_i \wedge B_j \subseteq B_i)$ for all $i \neq j$.
    Then $\lvert \mathcal{S} \rvert$ is finite.
\end{lemma}
\begin{proof}
    Suppose for contradiction that $\mathcal{S}$ is an infinite collection fulfilling these conditions.
    Since each $A_i$ contains elements from $[n]$ and $\lvert A_i \rvert = k$, there are only a finite number of possible definitions for the multiset $A_i$.
    Since $\mathcal{S}$ is infinite, it must contain some infinite subcollection $\mathcal{S}' = \{S_1', S_2', \ldots \} \subseteq \mathcal{S}$ such that $A_i' = A_j'$ for all $S_i', S_j' \in \mathcal{S}'$.
    Thus, for $\mathcal{S}$ to fulfill the conditions, we must have $B_i' \not \subseteq B_j'$ for all $S_i' \neq S_j'$.
    We will derive a contradiction by finding multisets in $\mathcal{X}_0 = \{B_i' \mid S_i' \in \mathcal{S}'\}$ such that one is a subset of the other.
    In fact, we show something much more general: $\mathcal{X}_0$ contains an infinite subcollection of equivalent multisets, i.e., multisets containing the same elements with the same multiplicities.

    We will use $\texttt{\#}(S, i)$ to denote the multiplicity of element $i \in [n]$ in multiset $S$.
    Choose any multiset $X \in \mathcal{X}_0$.
    For each $i \in [n]$, let $\mathcal{X}_1^{(i)} = \{B' \in \mathcal{X}_0 \setminus \{ X \} \mid \texttt{\#}(B', i) < \texttt{\#}(X, i)\}$ be the collection of multisets of $\mathcal{X}_0$ containing fewer instances of $i$ than $X$.
    Each $B' \in \mathcal{X}_0 \setminus \{X\}$ must exist in at least one of these collections since $X \not\subseteq B'$.
    But there are only $n$ collections $\mathcal{X}_1^{(1)}, \ldots, \mathcal{X}_1^{(n)}$, and $\mathcal{X}_0$ is infinite, so there must exist a collection $\mathcal{X}_1^{(i_1)}$ that is infinite.
    For each $0 \leq j < \texttt{\#}(X, i_1)$, let $\mathcal{X}_1^{(i_1,j)} = \{B' \in \mathcal{X}_1^{(i_1)} \mid \texttt{\#}(B', i_1) = j\}$ be the multisets in $\mathcal{X}_1^{(i_1)}$ containing exactly $j$ instances of element $i_1$.
    Once again, there are infinitely many multisets in $\mathcal{X}_1^{(i_1)}$ but only a finite range of multiplicities $j$, so at least one collection $\mathcal{X}_1^{(i_1,j)}$ is infinite.
    Call this one $\mathcal{X}_1$.

    Next define $\mathcal{X}_2$ in a similar way, but using $\mathcal{X}_1$ in place of $\mathcal{X}_0$.
    By our construction, every multiset in $\mathcal{X}_1$ contains the same number of instances of element $i_1 \in [n]$.
    When defining $\mathcal{X}_2$ then, we will have $\mathcal{X}_2^{(i_1)} = \varnothing$.
    Thus, there will be an element $i_2 \in [n]$ with $i_2 \neq i_1$ such that every multiset in $\mathcal{X}_2$ has the same number of instances of $i_2$.
    Since $\mathcal{X}_2 \subset \mathcal{X}_1$, every multiset in $\mathcal{X}_2$ has the same number of instances of both $i_1$ and $i_2$.
    If we continue to define the collections $\mathcal{X}_3, \dots, \mathcal{X}_n$ in this way, all multisets in each $\mathcal{X}_j$ will have the same numbers of instances of $i_1, \ldots, i_j$.
    Thus, for every $X, Y \in \mathcal{X}_n$ and $i \in [n]$, $\texttt{\#}(X,i) = \texttt{\#}(Y,i)$.
    But then $\mathcal{X}_n$ is an infinite subcollection of equivalent multisets in $\mathcal{X}_0$, a contradiction.
\end{proof}

\begin{figure}[t]
    \centering
    \includegraphics[width=\textwidth]{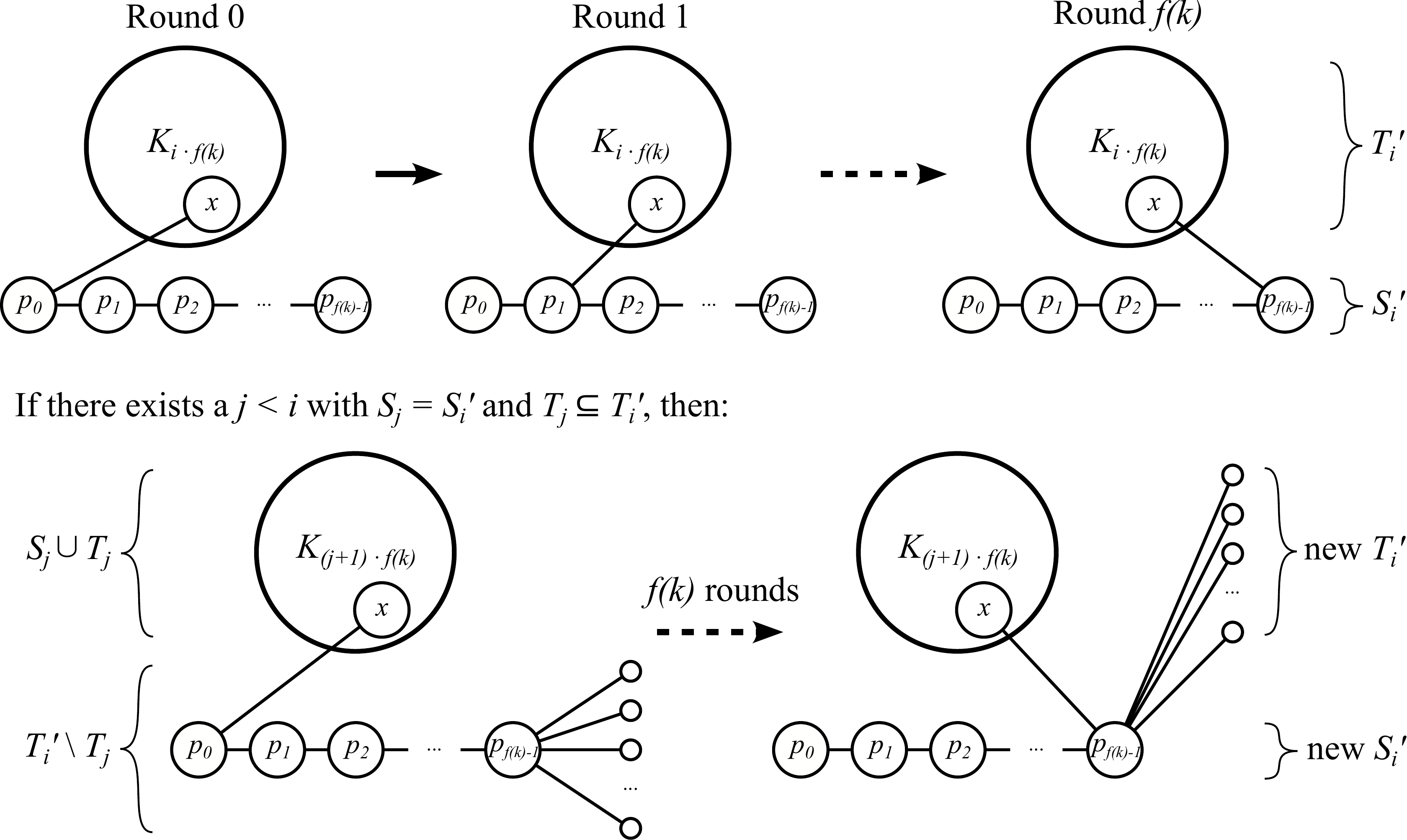}
    \caption{The time-varying graph structures used in the proof of Theorem~\ref{thm:stable-bound}.}
    \label{fig:stable-bound}
\end{figure}

We now prove our superconstant memory lower bound.

\begin{theorem} \label{thm:stable-bound}
    Any idle-start algorithm that solves broadcast with stabilizing termination for anonymous, synchronous, 1-interval connected dynamic networks must use $\omega(1)$ memory.
\end{theorem}
\begin{proof}
    Suppose for contradiction that $\alg$ is an idle-start algorithm with $\Theta(1)$ space complexity that solves broadcast with stabilizing termination.
    Then there is a constant $k$ such that $\alg$ never uses more than $k$ states.
    Note that throughout the rest of this proof, $f(k)$ is well-defined and has finite value by Lemma~\ref{lem:f}.
    We will show that the existence of this algorithm $\alg$ contradicts Lemma~\ref{lem:alpha} by constructing an infinite sequence of reachable configurations $(C_i)_{i=0}^\infty$ where each $C_i$ is a disjoint union of multisets $S_i \cup T_i$ satisfying:
    
    \begin{enumerate}
        \item $\lvert S_i \rvert = f(k)$,
        \item $\lvert T_i \rvert = i \cdot f(k)$, and
        \item $\neg (S_j = S_i \land T_j \subseteq T_i)$ for all $j \neq i$.
    \end{enumerate}

    Initially, the broadcaster is in some state $\beta_0$ and all other nodes are in some other state $\alpha_0 \neq \beta_0$.
    Define the initial configuration $C_0$ by letting $S_0 = \{\beta_0\} \cup \{\alpha_0\}^{f(k) - 1}$ and $T_0 = \varnothing$.
    Clearly, $C_0$ is reachable and satisfies the above conditions.
    
    Now consider any $i \geq 1$ and suppose configurations $C_0, \dots, C_{i-1}$ have already been defined; we inductively define $C_i$ as follows.
    Construct a time-varying graph $\mathcal{G}$ from two static components: the complete graph $K_{i \cdot f(k)}$ on $i \cdot f(k)$ nodes (including the broadcaster) and a path $p_0p_1 \cdots p_{f(k) - 1}$ (\figtext~\ref{fig:stable-bound}, top).
    In each round $t \in \{0, \ldots, f(k) - 1\}$, these components are connected by a single edge $\{x, p_t\}$, where $x \in K_{i \cdot f(k)}$ is some fixed node.
    It can be easily shown---as we did by induction in the proof of Theorem~\ref{thm:detection-impossibility}---that the nodes in $K_{i \cdot f(k)}$ send the same messages and transition to the same states in an execution of $\alg$ on $\mathcal{G}$ as they would in an execution of $\alg$ on $K_{i \cdot f(k)}$ by itself.
    By definition, these executions must stabilize on $K_{i \cdot f(k)}$ within $f(k)$ rounds.
    At this point, let $S_i'$ and $T_i'$ be the multisets of states of the nodes in the path and complete graph, respectively.

    If there is no $j \in \{0, \dots, i-1\}$ such that $S_j = S_i'$ and $T_j \subseteq T_i'$, then we define configuration $C_i$ as $S_i = S_i'$ and $T_i = T_i'$.
    Otherwise, if such a $j$ does exist, we extend the execution of $\alg$ on $\mathcal{G}$ as follows (\figtext~\ref{fig:stable-bound}, bottom).
    Arrange the nodes of $C_j = S_j \cup T_j$ as the complete graph $K_{(j+1) \cdot f(k)}$ and any $f(k)$ nodes from $T_i' \setminus T_j$ as a path $p_0p_1 \cdots p_{f(k) - 1}$; attach the remaining nodes to the far end $p_{f(k) - 1}$ of the path.
    Then repeat the same process as before, executing $\alg$ for $f(k)$ rounds when these components are connected in round $t$ by a single edge $\{x, p_t\}$, where $x \in K_{(j+1) \cdot f(k)}$ is any fixed node.
    Configuration $C_j$ is reachable by induction, and again all nodes of $K_{(j+1) \cdot f(k)}$ must send the same messages and transition to the same states in the execution of $\alg$ on $\mathcal{G}$ as they would in an execution of $\alg$ on $K_{(j+1) \cdot f(k)}$ alone.
    So, by definition of $f(k)$, all nodes in $K_{(j+1) \cdot f(k)}$ must be idle after these $f(k)$ rounds.
    Also, $p_{f(k) - 1}$ remains idle throughout this extended execution, so the initially idle non-path nodes attached to it are also idle at this time.
    Redefine $S_i'$ as the $f(k)$ path nodes and $T_i'$ as all other nodes, which as we've argued will all be idle.
    Again, if there is no $j \in \{0, \dots, i-1\}$ such that $S_j = S_i'$ and $T_j \subseteq T_i'$, then define configuration $C_i$ as $S_i = S_i'$ and $T_i = T_i'$; otherwise, extend the execution of $\alg$ on $\mathcal{G}$ by another $f(k)$ rounds as above.

    Suppose for contradiction that the condition for defining $C_i$ is never met and the execution of $\alg$ on $\mathcal{G}$ is extended forever.
    In every round of this execution, there is at least one idle node.
    So consider the execution of $\alg$ on a modified $\mathcal{G}$ containing an extra node that is attached to some idle node in each round.
    This execution is identical to the one on $\mathcal{G}$, but the extra node would never leave its initial state $\alpha_0$.
    Thus, there must also be non-idle nodes in every round, or else this execution stabilizes with uninformed nodes, contradicting the correctness of $\alg$.
    But then this infinite execution contains non-idle nodes in every round, contradicting the supposition that $\alg$ eventually stabilizes.
    Thus, multisets $S_i'$ and $T_i'$ with the desired condition will be found in finite time.
    Clearly, $C_i = S_i \cup T_i = S_i' \cup T_i'$ is reachable; we conclude by showing it satisfies the required conditions.
    
    \begin{enumerate}
        \item Every intermediate $S_i'$ is defined as the states of nodes in the path components which always comprise $f(k)$ nodes.
        Thus, $|S_i| = f(k)$.

        \item The time-varying graph $\mathcal{G}$ is defined on $|K_{i \cdot f(k)}| + |p_0p_1 \cdots p_{f(k) - 1}| = (i+1) \cdot f(k)$ nodes.
        Since $\lvert S_i \rvert = f(k)$, we have $\lvert T_i \rvert = i \cdot f(k)$.

        \item Consider any $j \neq i$.
        If $j < i$ then this condition must be met since the execution defining $C_i$ only stopped once it was true.
        Otherwise, if $j > i$, then by Condition 2, $\lvert T_j \rvert > \lvert T_i \rvert$ and thus $T_j \not \subseteq T_i$ is trivially true.
    \end{enumerate}

    Thus, the infinite sequence $(C_i)_{i=0}^{\infty}$ can be defined fulfilling all of the conditions above.
    But this contradicts Lemma~\ref{lem:alpha} stating that all such sequences are finite.
\end{proof}

\section{A Logspace Algorithm for Stabilizing Termination} \label{sec:stabilizing-algo}

In this section, we shift our attention from impossibilities and lower bounds to an idle-start algorithm called \countalg\ that solves broadcast with stabilizing termination in our anonymous, dynamic setting.
This algorithm uses $\bigo{\log n}$ memory---which is not far from the $\omega(1)$ lower bound of Theorem~\ref{thm:stable-bound}---and stabilizes in $\bigo{n}$ rounds which is worst-case asymptotically optimal.

\begin{algorithm}[t]
    \caption{\countalg\ for Node $v$}
    \label{alg:stabilizing}
    \begin{algorithmic}[1]
        \Statex \textbf{Initialization}. Set $\Current$ to $0$ and $\Maximum$ to $1$ if $v$ is the broadcaster and both variables to $-1$ otherwise.
        
        \medskip

        \Statex \textbf{Sending Messages}.
        \If {$\Current \neq -1$}
            \State \textsc{Send}: $\msg{msg}{\Current, \Maximum}$  \label{alg:stabilizing:msg}
        \EndIf

        \medskip

        \Statex \textbf{State Transitions}.        
        \If {$\Current \neq -1$}
            \State $\Current \gets \Current - 1$  \label{alg:stabilizing:decrement}
        \ElsIf {a message $\msg{msg}{c, m}$ was received}
            \If {$c = 0$}  \Comment{Initiate a new attempt.}
                \State $\Current \gets 2m$  \label{alg:stabilizing:reset-i}
                \State $\Maximum \gets 2m$  \label{alg:stabilizing:reset-f}
            \ElsIf {$c > 0$}  \Comment{Join the ongoing attempt.}
                \State $\Current \gets c - 1$  \label{alg:stabilizing:join-i}
                \State $\Maximum \gets m$  \label{alg:stabilizing:join-f}
            \EndIf
        \EndIf
	\end{algorithmic}
\end{algorithm}

At a high level, the \countalg\ algorithm (Algorithm~\ref{alg:stabilizing}) coordinates a sequence of broadcast attempts, each lasting twice as many rounds as its predecessor until one succeeds.
To facilitate these attempts, nodes store two values: \Current, the number of rounds remaining in the current attempt; and \Maximum, the total duration of the current attempt.
In each round, non-idle nodes involved in an ongoing attempt broadcast their \Current\ and \Maximum\ values to their neighbors and then decrement \Current.
Idle nodes that were previously not involved in the attempt but receive these messages will join in by setting their own \Current\ and \Maximum\ values accordingly.
This continues until messages are sent with $\Current = 0$, indicating the end of the current attempt.
If any idle node receives such a message, it detects that the broadcast should have gone on for longer.
It responds by initiating a new broadcast attempt whose duration is double the previous one.
These attempts continue until some attempt makes all nodes non-idle, at which point no node will initiate another attempt and the algorithm will stabilize.

Before analyzing this algorithm's correctness and complexity, we define some notation.
Let $v.\texttt{var}_t$ denote the value of variable \texttt{var} in the state of node $v$ at time $t$ (i.e., the start of round $t$).
In this notation, the \countalg\ algorithm initializes the broadcaster $b$ with $b.\Current_0 = 0$ and $b.\Maximum_0 = 1$ and all other nodes $v \neq b$ with $v.\Current_0 = v.\Maximum_0 = -1$.
Denote the set of non-idle nodes in round $t$ as $S_t = \{v \in V : v.\Current_t \neq -1\}$.
We begin our analysis by proving that all non-idle nodes share the same \Current\ and \Maximum\ values.

\begin{lemma} \label{lem:stablealg-synchronize}
    For all times $t$ and any non-idle node $v \in S_t$, we have $v.\Current_t = c_t$ and $v.\Maximum_t = m_t$, where $c_0 = 0$, $m_0 = 1$, and
    \[(c_{t+1}, m_{t+1}) = \left\{ \begin{array}{ll}
        (2m_t, 2m_t) & \text{if $c_t = 0$}; \\
        (c_t - 1, m_t) & \text{otherwise}.
    \end{array} \right.\]
\end{lemma}
\begin{proof}
    Argue by induction on $t$.
	Only the broadcaster $b$ is initially non-idle, so $c_0 = b.\Current_0 = 0$ and $m_0 = b.\Maximum_0 = 1$ by initialization.
	Now suppose the lemma holds up to and including some time $t \geq 0$ and let $c_t$ and $m_t$ be the unique values of $v.\Current_t$ and $v.\Maximum_t$ for all $v \in S_t$, respectively.
    Consider any node $v \in S_{t+1}$; if none exist, the lemma holds trivially.
    We have two cases:
    
    \begin{enumerate}
        \item $c_t = 0$. \label{lem:stablealg-synchronize:ct=0}
        Suppose to the contrary that $v \in S_t$; i.e., $v$ was also non-idle at time $t$.
        Then $v.\Current_t = c_t = 0$.
        Thus, $v$ must execute Line~\ref{alg:stabilizing:decrement} in round $t$, yielding $v.\Current_{t+1} = -1$ and becoming idle by time $t+1$, a contradiction.
        So $v$ was idle at time $t$ but became non-idle by time $t+1$, meaning it must have received one or more messages from non-idle nodes in round $t$.
        By Line~\ref{alg:stabilizing:msg} and the induction hypothesis, all of those messages are $\msg{msg}{c_t = 0, m_t}$.
        So $v$ must execute Lines~\ref{alg:stabilizing:reset-i}--\ref{alg:stabilizing:reset-f} in round $t$, yielding $v.\Current_{t+1} = v.\Maximum_{t+1} = 2m_t$.
        Our choice of $v \in S_{t+1}$ was arbitrary, so $c_{t+1} = m_{t+1} = 2m_t$.

        \item $c_t > 0$. \label{lem:stablealg-synchronize:ct>0}
        First suppose $v \in S_t$.
        Then $v.\Current_t = c_t > 0$, so $v$ executes Line~\ref{alg:stabilizing:decrement} in round $t$, yielding $v.\Current_{t+1} = c_t - 1$ and $v.\Maximum_{t+1} = m_t$ as claimed.
        Now suppose $v \not\in S_t$.
        To transition from idle to non-idle in round $t$, $v$ must receive one or more messages from non-idle nodes in round $t$.
        By Line~\ref{alg:stabilizing:msg}, all of those messages are $\msg{msg}{c_t > 0, m_t}$.
        So $v$ must execute Lines~\ref{alg:stabilizing:join-i}--\ref{alg:stabilizing:join-f}, yielding $v.\Current_{t+1} = c_t - 1$ and $v.\Maximum_{t+1} = m_t$.
        Our choice of $v \in S_{t+1}$ was arbitrary, so $c_{t+1} = c_t - 1$ and $m_{t+1} = m_t$.
        \qedhere
    \end{enumerate}
\end{proof}

Using the $c_t$ and $m_t$ values defined in Lemma~\ref{lem:stablealg-synchronize}, we next show that a broadcast attempt lasting $k$ rounds either informs all nodes and stabilizes or involves at least $k+1$ non-idle nodes before initiating a new attempt.

\begin{lemma} \label{lem:stablealg-progress}
	If $k := c_t = m_t > 0$, then either (1) $S_{t+k} = V$ and $S_{t+k+1} = \varnothing$, or (2) $|S_{t+k}| \geq k+1$ and $c_{t+k+1} = m_{t+k+1} = 2k$.
\end{lemma}
\begin{proof}
    Consider any time $t$ at which $k := c_t = m_t > 0$, the start of a new broadcast attempt by the set of nodes $S_t$.
    First suppose that for some $0 \leq i \leq k$, we have $S_{t+i} = V$; i.e., all nodes are non-idle.
    By Lemma~\ref{lem:stablealg-synchronize}, we have $c_{t+i} = c_t - i$.
    Since all nodes $v$ are non-idle at time $t+i$ and thus have $v.\Current_{t+i} = c_{t+i} = c_t - i$, they all execute Line~\ref{alg:stabilizing:decrement} in round $t+i$ by decrementing $v.\Current$.
    If $i < k$, then $v.\Current_{t+i+1} = c_t - i - 1 \neq -1$, so all nodes remain non-idle and the process repeats; otherwise, if $i = k$, then $v.\Current_{t+i+1} = c_t - k - 1 = -1$.
    This renders all nodes idle at time $t+k+1$, so Case~\ref{lem:stablealg-synchronize:ct=0} has occurred.

    Now suppose that $S_{t+i} \neq V$ for all $0 \leq i \leq k$; i.e., there is at least one idle node throughout the attempt.
    We argue by induction on $0 \leq i \leq k$ that $|S_{t+i}| \geq |S_t| + i$; i.e., at least one idle node becomes non-idle in each round.
    The $i = 0$ case holds trivially, so suppose the claim holds up to and including some $i < k$.
    Since $c_{t+i} = c_t - i = k - i$ by Lemma~\ref{lem:stablealg-synchronize}, any node in $S_{t+i}$ must remain non-idle until time $t + i + c_{t+i} = t + k$.
    So $|S_{t+i+1}| \geq |S_{t+i}|$.
    If the induction hypothesis is in fact a strict inequality, we are done:
    \[|S_{t+i}| > |S_t| + i \quad \Rightarrow \quad |S_{t+i+1}| \geq |S_{t+i}| \geq |S_t| + i + 1.\]
    So suppose instead that $|S_{t+i}| = |S_t| + i$.
    There must exist non-idle nodes at time $t+1$ since $c_{t+1} = c_t - 1 > -1$ by Lemma~\ref{lem:stablealg-synchronize}, and there must exist idle nodes at time $t+1$ by supposition.
    Thus, since the dynamic network is 1-interval connected, there must be some idle node $v \in V \setminus S_{t+i}$ that receives a message from a non-idle node in round $t+1$, causing $v$ to become non-idle (Lines~\ref{alg:stabilizing:reset-i}--\ref{alg:stabilizing:reset-f} or~\ref{alg:stabilizing:join-i}--\ref{alg:stabilizing:join-f}).
    By the induction hypothesis, $|S_{t+i+1}| \geq |S_{t+i}| + 1 = |S_t| + i + 1$.

    By this induction argument, we have $|S_{t+k}| \geq |S_t| + k \geq k + 1$.
    By Lemma~\ref{lem:stablealg-synchronize}, we have that $c_{t+k} = c_t - k = 0$, and with another application of the same lemma, we conclude that $c_{t+k+1} = m_{t+k+1} = 2m_{t+k} = 2m_t = 2k$.
    So Case~\ref{lem:stablealg-synchronize:ct>0} has occurred.
\end{proof}

Our algorithm's correctness follows from the previous lemma and its time and space complexities are obtained with straightforward counting arguments.

\begin{theorem} \label{thm:stabilizing}
    \countalg\ (Algorithm~\ref{alg:stabilizing}) correctly solves broadcast with stabilizing termination from an idle start in $\bigo{n}$ rounds and $\bigo{\log n}$ space for anonymous, synchronous, 1-interval connected dynamic networks.
\end{theorem}
\begin{proof}
    By Lemma~\ref{lem:stablealg-synchronize}, we have $c_1 = m_1 = 2$, allowing us to apply Lemma~\ref{lem:stablealg-progress}.
    But suppose to the contrary that this and all subsequent applications of the lemma result in Case~\ref{lem:stablealg-synchronize:ct>0}---where a new attempt is initiated with double the duration---and not Case~\ref{lem:stablealg-synchronize:ct=0}, where all nodes are informed ($S_{t+k} = V$) and the algorithm stabilizes ($S_{t+k+1} = \varnothing$).
    Then there exists an attempt of duration $k \geq n$, which, by Lemma~\ref{lem:stablealg-progress}, ends with $|S_{t+k}| \geq k + 1 \geq n + 1$ non-idle nodes.
    But there are only $n$ nodes in the network, a contradiction.
    So \countalg\ must inform all nodes and stabilize in finite time.

    It remains to bound runtime and memory.
    Each time a new broadcast attempt is initiated, the $\Maximum$ variable is doubled and the algorithm runs for another $\Maximum$ rounds.
    As we already showed, once $\Maximum$ reaches or exceeds $n$, the subsequent attempt will inform all nodes and stabilize.
    Thus, $\Maximum$ doubles at most $\lceil \log_2 n\rceil$ times, meaning \countalg\ stabilizes in at most $\sum_{i=0}^{\lceil \log_2 n\rceil} 2^i = \bigo{2^{1+\lceil \log_2 n\rceil} - 1} = \bigo{n}$ rounds.
    This analysis also shows that the largest attainable $\Maximum$ value before its final doubling is $n - 1$, so $\Maximum \leq 2(n-1) = \bigo{n}$.
    Since $-1 \leq \Current \leq \Maximum$, we also have $\Current = \bigo{n}$, implying that \countalg\ has $\bigo{\log n}$ space complexity.
\end{proof}

\section{Conclusion} \label{sec:conclude}

This paper investigated what memory is necessary for anonymous, synchronous, 1-interval connected dynamic networks to deterministically solve broadcast with some termination conditions.
We considered both \textit{termination detection} where the broadcaster must eventually declare that every node has been informed and \textit{stabilizing termination} where nodes must eventually stop sending messages.
Combining our results with the established literature, we now know the following about this problem:

\begin{itemize}
    \item \textit{Termination Detection}.
    Regardless of memory, broadcast with termination detection is impossible for idle-start algorithms (Theorem~\ref{thm:detection-impossibility}) and for non-idle-start algorithms when the number of broadcasters is unknown (Theorem~\ref{thm:detect-multileader}).
    Any (non-idle-start) algorithm solving broadcast with termination detection must use $\Omega(\log n)$ memory per node (Theorem~\ref{thm:detection-bound}).
    The best known space complexity for this problem follows from Di Luna and Viglietta's history trees algorithm which uses $\bigo{n^3\log n}$ memory in the worst case~\cite{DiLuna2023-briefannouncement}.

    \item \textit{Stabilizing Termination}.
    Any idle-start algorithm solving broadcast with stabilizing termination must use $\omega(1)$ memory per node (Theorem~\ref{thm:stable-bound}).
    Building evidence against O'Dell and Wattenhofer's conjecture~\cite{ODell2005-informationdissemination}, this problem is solvable with logarithmic memory under standard synchrony: \countalg\ is a $\bigo{\log n}$ memory, linear time algorithm achieving stabilizing termination without identifiers or knowledge of $n$ (Theorem~\ref{thm:stabilizing}).
\end{itemize}

For stabilizing termination, our $\omega(1)$ memory bound holds only for idle-start algorithms and our $\bigo{\log n}$ memory \countalg\ algorithm happens to be idle-start.
In the non-idle-start regime where non-broadcaster nodes can send messages from their initial states, can we obtain a sublogarithmic space algorithm?
Can any lower bound be shown?
The contradiction at the heart of our lower bound technique for idle-start algorithms identified configurations $C_i$ that are reachable with or without an extra uninformed node.
In the non-idle-start case, however, this extra uninformed node may send new and unaccounted for messages and we can no longer guarantee $C_i$ will still be reached, requiring a different approach.

The $\Omega(\log n)$ and $\bigo{n^3\log n}$ memory bounds for termination detection leave open a significant gap for further improvement.
Our logarithmic lower bound shows that termination detection requires enough memory to count to $n$, and indeed there is a straightforward solution for termination detection if (an upper bound on) $n$ can be obtained: simply wait for $n$ rounds after broadcasting information for the first time and then declare all nodes have been informed.
Can broadcast with termination detection be achieved without solving exact counting?
If not, what approaches could yield algorithms for exact counting that are more space-efficient than history trees?
Viglietta recently proposed the existence of logspace counting algorithms as an open problem unlikely to be solved by history trees~\cite{Viglietta2024-historytrees}; we are unsure such algorithms exist at all, as we suspect our $\Omega(\log n)$ bound can be improved.

Finally, we note that all impossibility results and lower bounds in this paper apply also to any problem broadcast reduces to, such as exact counting.
Thus, these results and any future improvements shed important light on the requirements of nontrivial terminating computation in anonymous dynamic networks.

\bibliographystyle{plainurl}
\bibliography{ref}

\end{document}